\documentclass[a4paper, 11pt]{article}

\usepackage[english]{babel}
\usepackage{amsmath, amssymb, amsfonts, bm}
\usepackage{mathrsfs}
\usepackage{bbm}
\usepackage{float}
\usepackage{color,graphicx,epstopdf}
\usepackage{tikz,pgf}
\usetikzlibrary{arrows,shapes,shadows,petri}
\usepackage{dsfont}
\usepackage{natbib}

\usepackage{indentfirst}
\usepackage{subcaption}

\DeclareMathOperator{\tr}{Tr}

\newcommand{\rmp}{\ensuremath{\mathrm{p}}}
\newcommand{\HH}{\ensuremath{\mathcal{H}}}
\newcommand{\Hn}{\ensuremath{\mathcal{H}^{\otimes n}}}
\newcommand{\bbR}{\mathbb{R}}
\newcommand{\bbC}{\mathbb{C}}
\newcommand{\bbP}{\mathbb{P}}

\newcommand{\calL}{\mathcal{L}}
\newcommand{\calD}{\mathcal{D}}

\newcommand{\scrS}{\mathscr{S}}
\newcommand{\scrL}{\mathscr{L}}
\newcommand{\sfH}{\mathsf{H}}
\newcommand{\sfP}{\mathsf{P}}
\newcommand{\sfM}{\mathsf{M}}

\newcommand{\sfV}{\mathsf{V}}
\newcommand{\bra}[1]{\ensuremath{|#1\rangle}}
\newcommand{\ket}[1]{\ensuremath{\langle#1|}}
\newcommand{\braket}[2]{\ensuremath{|#1\rangle\langle#2|}}
\newcommand{\ketbra}[2]{\ensuremath{\langle#1|#2\rangle}}

\newcommand{\btoi}[1]{\ensuremath{\lfloor #1\rfloor}}
\newcommand{\itob}[1]{\ensuremath{\lceil #1\rceil}}
\newcommand{\vr}{\mathbf{r}}

\newcommand{\vL}{\mathbf{L}}
\newcommand{\vD}{\mathbf{D}}
\newcommand{\vP}{\mathbf{P}}

\newcommand{\vW}{\mathbf{W}}
\newcommand{\vx}{\mathbf{x}}

\def\P{\ensuremath{\mathsf{P}}}

\newcommand{\M}{\ensuremath{\mathsf{M}}}

\newcommand*\mcapinn[2]{\vcenter{\hbox{$\mathsurround=0pt
 \ifx\displaystyle#1\textstyle\else#1\fi\bigcap$}}}

\newcommand*\mcupinn[2]{\vcenter{\hbox{$\mathsurround=0pt
 \ifx\displaystyle#1\textstyle\else#1\fi\bigcup$}}}

\DeclareFontFamily{OT1}{pzc}{}
\DeclareFontShape{OT1}{pzc}{m}{it}{<-> s * [1.200] pzcmi7t}{}
\DeclareMathAlphabet{\mathpzc}{OT1}{pzc}{m}{it}

\newcommand{\pzcV}{\mathpzc{V}}

\graphicspath{{./images/}}

\newtheorem{theorem}{Theorem}
\newtheorem{definition}{Definition}

\newtheorem{lemma}{Lemma}
\newtheorem{remark}{Remark}
\newtheorem{proposition}{Proposition}

\newenvironment{proof}[1][Proof]{\noindent\textbf{#1.} }{\hfill \rule{0.5em}{0.5em}}

\begin{document}
	
	\title{\bf Measurement-Induced Boolean Dynamics for Open Quantum Networks}
	
	\author{Hongsheng Qi\thanks{Key Laboratory of Systems and Control, Institute of Systems Science, Academy of Mathematics and Systems Science, Chinese Academy of Sciences, Beijing 100190, China; Research School of Electrical, Energy and Materials Engineering,  The Australian National University, Canberra 0200, Australia. E-mail: qihongsh@amss.ac.cn.}, Biqiang Mu\thanks{Key Laboratory of Systems and Control, Institute of Systems Science, Academy of Mathematics and Systems Science, Chinese Academy of Sciences, Beijing 100190, China. E-mail: bqmu@amss.ac.cn.}, Ian R. Petersen\thanks{Research School of Electrical, Energy and Materials Engineering,  The Australian National University, Canberra 0200, Australia. E-mail: ian.petersen@anu.edu.au.}, Guodong Shi\thanks{Australian Center for Field Robotics, School of Aerospace, Mechanical and Mechatronic Engineering, The University of Sydney, NSW 2008, Australia. E-mail: guodong.shi@sydney.edu.au.}}
	
	\date{}

	\maketitle
	
	\begin{abstract}
		In this paper, we study the recursion of measurement outcomes for open quantum networks under sequential measurements. Open quantum networks are networked quantum subsystems (e.g., qubits) with the state evolutions described by a continuous Lindblad master equation. When measurements are performed sequentially along such continuous dynamics, the quantum network states undergo random jumps and the corresponding measurement outcomes can be described by a vector of probabilistic Boolean variables. The induced recursion of the Boolean vectors forms a probabilistic Boolean network. First of all, we show that the state transition of the induced Boolean networks can be explicitly represented through realification of the master equation. Next, when the open quantum dynamics is relaxing in the sense that it possesses a unique equilibrium as a global attractor, structural properties including absorbing states, reducibility, and periodicity for the induced Boolean network are direct consequences of the relaxing property. Particularly, we show that generically, relaxing quantum dynamics leads to irreducible and aperiodic chains for the measurement outcomes. Finally, we show that for quantum consensus networks as a type of non-relaxing open quantum network dynamics, the communication classes of the measurement-induced Boolean networks are encoded in the quantum Laplacian of the underlying interaction graph. 	\end{abstract}



\section{Introduction}

It is widely believed that our future information infrastructures will be built on quantum technologies, where computing and communication take place over states of quantum systems \citep{Nielsen}. Quantum systems are fundamentally different from classical states from the following aspects: quantum states are vectors in finite or infinite dimensional Hilbert spaces; isolated quantum systems exhibit closed dynamics described by Schr\"odinger equations; performing measurements over a quantum system yields random outcomes and the quantum system collapses to states corresponding to the measurement outcomes. On one hand, these distinct properties of quantum systems empower quantum computing and communication to a degree that classical systems cannot achieve. On the other hand, the creation, measurement, preservation, and manipulation of quantum systems become very difficult, especially at a large scale \citep{Wiseman-Milburn-2009}. 

Quantum systems may also be exposed to external environments with whom they form closed quantum dynamics. 
When a Markovian approximation is applied
under the assumption of a short environmental correlation time permitting the neglect of memory
effects \citep{Breuer-Petruccione-2002}, a master equation can be used to describe
the quantum state evolution \citep{Lindblad1976}, forming open quantum systems. The interaction between an open quantum system and its driven environment essentially define input-output feedback channels \citep{ian2008} within the overall system. Open quantum systems arise from a large context of quantum systems and quantum engineering \citep{Ticozzi-Viola-2008,Ticozzi-Viola-2009,WangPRA,Altafini2002,Altafini-Ticozzi-2012,Wang2012}. Particularly, if we connect a network of quantum subsystems such as qubits by a series of local environments, an open quantum network is obtained. Open quantum networks have proven to be a resource for universal quantum computing \citep{NaturePhysiscs}, and a way of realizing quantum consensus and synchronization at the quantum level~\citep{Ticozzi2015,shi2016}.

In practice, knowledge of any quantum system has to come from some form of quantum measurements, where only random outcomes can be obtained.
Repeated quantum measurements are the major approach for identifying unknown quantum states, where independent measurements are carried out for identical copies of the states for identification.
For the manipulation of quantum states, sequential quantum measurements can be performed as a way of realizing feedforward control with a static quantum plant \citep{Rabitz-PRA}, or measurement feedback control with closed quantum dynamics \citep{Belavkin1979,QubitFeedback2014}. Due to the complexity of the open quantum dynamics, the role of sequential measurements have not been quite understood in the literature even only for the recursion of the measurement outcomes.

In this paper, we study open quantum networks of qubits with sequential measurements. When measurements are performed sequentially along the continuous dynamics, the quantum network states undergo random jumps and the measurement outcomes are naturally described by a vector of random Boolean variables, forming a probabilistic Boolean network \citep{Probabilistic-Boolean-Network}. The induced recursion of the Boolean networks defines a Markov chain, which is governed both by the master equation of the continuous quantum dynamics, and the basis of the network measurement. We establish a clear and explicit representation for the state transition of the random measurement outcomes from realification of the master equation. Moreover, for relaxing and non-relaxing quantum network dynamics, respectively, we establish the following results.
\begin{itemize}
\item[(i)] When the open quantum dynamics is relaxing, i.e., it possesses a unique equilibrium that is globally asymptotically stable, structural properties including absorbing states, reducibility, and periodicity for the induced Boolean networks are established directly from the relationship between the master equation and the measurement basis. Particularly, we show that as a generic property, relaxing quantum dynamics leads to irreducible and aperiodic chains for the measurement outcomes. 
\item[(ii)]We show that for quantum consensus networks as a special type of non-relaxing open quantum network dynamics, the communication classes of the measurement-induced Boolean networks are fully encoded in the quantum Laplacian of the underlying interaction graph.
\end{itemize}

The remainder of the paper is organized as follows. In Section \ref{Sec2}, we present some preliminary theory and notations for the sake of achieving self-contained presentation. In Section \ref{sec3}, we introduce the quantum network model, the resulting hybrid quantum network dynamics, and the definition of the induced probabilistic Boolean network from the measurement outcomes. In Section \ref{sec4}, we establish the representation of the state transition of the Boolean network. Then in Sections \ref{sec5} and \ref{sec6}, we present the results for relaxing and non-relaxing quantum dynamics, respectively. Finally some concluding remarks are presented in Section \ref{sec7}. 

\section{Preliminaries}\label{Sec2}
In this section, we introduce some concepts and theory from graph theory \citep{godsil}, open quantum systems \citep{Nielsen}, and Markov chains \citep{durrett}.

\subsection{Graph Theory Essentials}
A simple undirected graph $\mathrm
{G}=(\mathrm {V}, \mathrm {E})$ consists of a finite set
$\mathrm{V}=\{1,\dots,n\}$ of nodes and an edge set
$\mathrm {E}$, where an element $e=\{i,j\}\in\mathrm {E}$ denotes an
{\it edge} between two distinct nodes $i\in \mathrm{V}$ and $j\in\mathrm{V}$. Two nodes $i,j\in\mathrm{V}$ are said to be {\it adjacent}
if $\{i,j\}$ is an edge in $\mathrm{E}$. The number of adjacent nodes of $v$ is called its degree, denoted ${\rm deg}(v)$.
The nodes that are adjacent with a node $v$ as well as itself are called its neighbors. A graph $\mathrm
{G}$ is called to be {\it regular} if all the nodes have the same degree. A path between two vertices $v_1$ and $v_k$ in $\mathrm{G}$ is a sequence of distinct nodes
$
v_1v_2\dots v_{k}$
such that for any $m=1,\dots,k-1$, there is an edge between $v_m$ and $v_{m+1}$. A pair of distinct nodes $i$ and $j$
is called to be {\it reachable} from each other if there is a path between them. A node is always assumed to be reachable
from itself. We call graph $\mathrm{G}$ {\it connected} if every pair of distinct nodes in $\mathrm{V}$ are reachable from each other.
 A subgraph of $\mathrm{G}$ associated with node set $\mathrm{V}^\ast \subseteq \mathrm{V}$, denoted as $\mathrm{G}|_{\mathrm{V}^\ast}$,
 is the graph $(\mathrm{V}^\ast, \mathrm{E}^\ast)$, where $\{i,j\}\in \mathrm{E}^\ast$ if and only if $\{i,j\}\in \mathrm{E}$ for $i,j\in \mathrm{V}^\ast$.
 A connected component (or just component) of $\mathrm
{G}$ is a connected subgraph induced by some $\mathrm{V}^\ast \subseteq \mathrm{V}$, which is connected to no additional nodes in $\mathrm
{V}\setminus \mathrm{V}^\ast$.

The (weighted) Laplacian of $\mathrm
{G}$, denoted $L(\mathrm
{G})$, is defined as \citep{Magnus}
$$
L(\mathrm
{G})=D(\mathrm
{G})-A(\mathrm
{G}),
$$
where $A(\mathrm{G})$ is the $n\times n$ matrix given by $[A(\mathrm{G})]_{kj}=[A(\mathrm{G})]_{jk}=
a_{kj}$ for some $a_{kj}>0$ if $\{k,j\} \in \mathrm{E}$ and $[A(\mathrm{G})]_{kj}=0$ otherwise, and $D(\mathrm
{G})={\rm diag}(d_1,\dots,d_N)$ with $d_k=\sum_{j=1,j\neq k}^N [A(\mathrm{G})]_{kj}$. It is well known that $L(\mathrm
{G})$ is always positive semi-definite, and ${\rm rank}(L(\mathrm
{G}))=n-C_\ast(\mathrm{G})$, where $C_\ast(\mathrm{G})$ denotes the number of connected components of $\mathrm{G}$.

\subsection{Open Quantum Systems}
\subsubsection{Quantum States}The state space of any isolated quantum system is a complex vector space with inner product, i.e., a Hilbert space $\HH_N\simeq\bbC^N$ with $N\ge 2$.
The system state is described by a unit vector in $\HH_N$
denoted by $\bra{\varphi}$, where $\bra{\cdot}$ is known as the Dirac notion for vectors representing (pure) quantum states. The states of a composite quantum system of two subsystems with state space $\HH_A$ and $\HH_B$, respectively, are complex linear combinations of $\bra{\varphi_A}\otimes \bra{\varphi_B}$, where $\bra{\varphi_A}\in\HH_A$, $\bra{\varphi_B}\in\HH_B$. For any $\bra{p}, \bra{q} \in \HH_N$, we use the notation $\bra{p} \ket{q}$ to denote the operator over $\HH_N$ defined by
$$
\big(\bra{p} \ket{q} \big) \bra{\eta}= \Big \langle \bra{q}, \bra{\eta}\Big\rangle \bra{p}, \ \ \forall \bra{\eta} \in \HH_N,
$$
where $\big\langle \cdot, \cdot \big\rangle$ represents the inner product that the Hilbert space $\HH_N$ is equipped with.
In standard quantum mechanical notation, the inner product $\big \langle \bra{p}, \bra{q}\big\rangle $ is denoted as $\ketbra{p}{q}$.

 Quantum states as ensembles of pure states can also be described by a positive Hermitian density operator over the space $\HH_N$ (or density matrix) $\rho$ satisfying $\text{Tr}(\rho)=1$. 
 
 \subsubsection{Quantum Measurements}
 Let $\calL(\HH_N)$ be the space of linear operators over $\HH_N$. For a quantum system associated with state space $\HH_N$, a projective measurement is described by an observable $\M$, which is a Hermitian operator in $\calL(\HH_N)$. The observable $\M$ has a spectral decomposition in the form of
$$
\M=\sum\limits_{m=0}^{N-1} \lambda_m\P_m,
$$
where $\P_m$ is the projector onto the eigenspace of $\M$ with eigenvalue $\lambda_m$. The possible outcomes of the measurement correspond to the eigenvalues $\lambda_m$, $m=0,\dots,N-1$ of the observable. Upon measuring the state $\bra{\varphi}$ (or $\rho$), the probability of getting result $\lambda_m$ is given by
$
p(\lambda_m)=\ket{\varphi}\P_m\bra{\varphi}
$ (or ${\rm Tr}(\rho \P_m)$).
Given that outcome $\lambda_m$ occurred, the state of the quantum system immediately after the measurement is
$\frac{\P_m\bra{\varphi}}{\sqrt{p(\lambda_m)}}$ (or $\frac{\P_m\rho\P_m}{p(\lambda_m)}$). 

\subsubsection{Master Equations}
The time evolution of the state $\bra{\varphi(s)}\in\HH_N$ of a closed quantum system is described by a Schr\"odinger equation. A quantum systems may also interact with external environments, who are quantum systems by themselves, and the composite system generated by the system and the environments form an isolated (closed) quantum system. When a Markovian approximation can be applied
under the assumption of a short environmental correlation time permitting the neglect of memory
effects \citep{Breuer-Petruccione-2002}, the time evolution of an open quantum system can be described by a Lindblad master equation as
\begin{align}\label{eq:lindblad}
\dot\rho(s) = -\imath[\sfH,\rho(s)]+\scrL_D(\rho(s))\equiv \scrL(\rho(s)),
\end{align}
where $\imath=\sqrt{-1}$, $\sfH$ is the Hamiltonian as a Hermitian operator
over $\HH_N$, and $\scrL_D(\rho)=\sum_{d}\mathcal{D}[\sfV
_d]\rho$ is the Lindblad operator from environments. The $\sfV_d$ are operators over $\HH_N$, and
\begin{align}
\calD[\sfV_d]\rho(s)=\sfV_d\rho(s) \sfV^\dag_d-\frac{1}{2}[\sfV_d^\dag \sfV_d\rho(s)+\rho(s) \sfV_d^\dag \sfV_d]
\end{align}
where $(\cdot)^\dag$ represents the Hermitian conjugate. Let $e^{\scrL s}$ be the quantum dynamical map governed by \eqref{eq:lindblad} for $s\ge 0$. Then $\{e^{\scrL s}\}_{s\ge 0}$ forms a semigroup \citep{Breuer-Petruccione-2002}. 

\begin{definition}\label{def1}
	The semigroup $\{e^{\scrL s}\}_{s\ge 0}$ is relaxing if there exists a unique (steady) state $\rho_\star$, such that $e^{\scrL s}(\rho_\star)=\rho_\star$ for all $s$, and 
	$$
	\lim_{s\rightarrow\infty}e^{\scrL s}(\rho(0))=\rho_\star,
	$$
	for all $\rho(0)$.
\end{definition}

Note that from the dynamical system perspective, the semigroup $\{e^{\scrL s}\}_{s\ge 0}$ being relaxing means the system \eqref{eq:lindblad} has a unique equilibrium that is globally asymptotically stable.

\subsection{Markov Chains}

Let $\scrS$ be the finite set $\{1,2,\dots,m\}$. Let $\vP$ be an $m\times m$ non-negative matrix with $\sum_{j=1}^{m}[\vP]_{ij}=1$ for $i\in\scrS$. A stochastic process $\{\vx(t)\}_{t=0}^{\infty}$ with state space $\scrS$ is called a homogeneous Markov chain with transition matrix $\vP$, if there holds
$$
\bbP(\vx(t+1)|\vx(0),\dots,\vx(t))=\bbP(\vx(t+1)|\vx(t))
$$
for $t=0,1,2,\dots$, and 
$$
\bbP(\vx(t+1)=j|\vx(t)=i)=[\vP]_{ij},
$$
for all $i,j\in\scrS$. 

Let row vector $\pi_0$ be the initial distribution of the time homogeneous Markov chain $\{\vx(t)\}_{t=0}^{\infty}$ with $[\pi_0]_i=\bbP(\vx(0)=i)$ and $\sum_{i\in\scrS}[\pi_0]_i=1$.
Let $\pi_t$ denote the distribution of the chain at time $t$, i.e., $[\pi_t]_i=\bbP(\vx(t)=i)$. Then there holds $$
\begin{aligned}
\left[\pi_{t+1}\right]_j = \sum_{i=1}^{m} [\pi_{t}]_i [\vP]_{ij},
\end{aligned}
$$
or in a compact form, 
$
\pi_{t+1} = \pi_t\vP.
$

A time homogeneous Markov chain $\{\vx(t)\}_{t=0}^{\infty}$ with the state space $\scrS$ is called irreducible if there exists an integer $l\geq 1$ such that $[\vP^l]_{ij}>0$ for any $i,j\in \scrS$.
The period $d(i)$ of a state $i\in \scrS$ is defined as the greatest common divisor of all $l$ that satisfy $[\vP^l]_{ii}>0$ and $\{\vx(t)\}_{t=0}^{\infty}$ is called aperiodic if all the states have period one. If the chain is both irreducible and aperiodic, then there exist a row vector $\pi_\ast$ satisfying 
$$
\pi_\ast=\lim_{t\to \infty}\pi_0\vP^t
$$
for all initial distribution $\pi_0$. In that case $\pi_\ast$ is termed the stationary distribution of the Markov chain. 

\section{Problem Definition}\label{sec3}
\subsection{Qubit Networks}

Qubit is the simplest quantum system whose state space is a two-dimensional Hilbert space $\HH$ ($:=\HH_2$). 
Consider a quantum network with $n$ qubits, which are indexed by $\mathrm{V}=\{1,\dots,n\}$. The state space of the $n$-qubit network is denoted as $\Hn=\underbrace{\HH\otimes\cdots\otimes\HH}_n$ (i.e., $\HH_{2^n}$). The density operator of the $n$-qubit network is denoted as $\rho$. Let there be an observable (or a projective measurement) for a single qubit as
$$
\sfM=\lambda_0\sfP_0+\lambda_1\sfP_1,
$$
where $\sfP_m=\braket{v_m}{v_m}$ is the projector onto the eigenspace generated by the eigenvector $\bra{v_m} \in \HH_{2}$ with eigenvalue $\lambda_m$, $m\in\{0,1\}$. Then $\sfM^{\otimes n}$ is an observable of the $n$-qubit network. 

\subsection{Open Quantum Networks with Sequential Measurements}

Consider the continuous time horizon for $s\in[0,\infty)$. Let the open quantum network state $\rho(s)$ be measured along $\sfM^{\otimes n}$ from $s=0$ periodically with a period $\tau$.
To be precise, $\rho(s)$ satisfies the following hybrid dynamics
\begin{subequations}\label{eq:hybrid}
\begin{align}\label{eq:hybrid.dynamic}
&\dot\rho(s) = \mathscr{L}(\rho(s)),\quad s\in[t\tau,(t+1)\tau),\\
\label{eq:hybird.initial}
&\rho(t\tau) = \rho_\rmp((t\tau)^-),
\end{align}
\end{subequations}
for $t=0,1,2,\dots$. Here
 $\rho((t\tau)^-)$ represents the quantum network state right before $t\tau$ along \eqref{eq:hybrid.dynamic} starting from $\rho(t\tau)$, and $\rho_\rmp((t\tau)^-)$ is the post-measurement state of the network when a measurement is performed at time $s=t\tau$ along $\sfM^{\otimes n}$, respectively.

We further introduce 
$$
\begin{aligned}
\xi(t)&:=\rho((t\tau)^-),\\
\xi_\rmp(t)&:=\rho_\rmp((t\tau)^-),
\end{aligned}
$$
for the pre- and post-measurement network states at the $t$-th measurement.

\subsection{Induced Boolean Networks}
The measurement $\M^{\otimes n}$ measures the individual qubit states of the entire network, which yields $2^n$ possible outcomes $[\lambda_{m_1},\dots,\lambda_{m_n}],{m_j}\in\{0,1\}, j=1,\dots,n$. 
We use the Boolean variable $x_i(t)\in\{0,1\}$ to represent the measurement outcome at qubit $i$ for step $t$, where $x_i(t)=0$ corresponds to $\lambda_0$ and $x_i(t)=1$ corresponds to $\lambda_1$. We can further define the $n$-dimensional random Boolean vector
$$
\mathbf{x}(t)=[x_1(t),\cdots,x_n(t)]\in \{0,1\}^{n}
$$
as the outcome of measuring $\xi(t)$ 
under $\M^{\otimes n}$ at step $t$. 

Clearly, $\{\vx(t)\}_{t=0}^{\infty}$ forms a Markov chain as the distribution of $\mathbf{x}(t+1)$ is fully determined by $\xi_\rmp(t+1)$, which depends only on $\mathbf{x}(t)$, e.g., Figure~\ref{fig:induced.boolean}. The $\mathbf{x}(t),t=0,1,2,\dots$ therefore falls to the category of classical probabilistic Boolean networks \citep{Probabilistic-Boolean-Network}. 

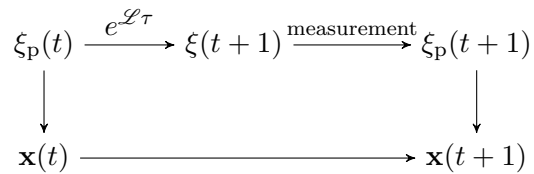
\begin{figure}[!htb]
	\centering
	\begin{tikzpicture}[->,>=stealth',auto,node distance=2.5cm]
	\tikzstyle{state}=[semithick,minimum size=6mm]
	\node[state] (q2) {$\xi_\rmp(t)$};
	\node[state,right of=q2] (q3) {$\xi(t+1)$};
	\node[right of=q3,xshift=1.8em] (q4) {$\xi_\rmp(t+1)$};
	\node[state,below of=q2,yshift=1cm] (x1) {$\mathbf{x}(t)$};
	\node[state,below of=q4,yshift=1cm] (x2) {$\mathbf{x}(t+1)$};
	\path (q2) edge node[above]{$e^{\scrL \tau}$} (q3);
	\path (q3) edge node[above]{\scriptsize measurement} (q4);
	\path[->] (q2) edge (x1);
	\path[->] (q4) edge (x2);
	\path (x1) edge node[above]{} (x2);
	\end{tikzpicture}
	\caption{Induced Boolean network dynamics from the sequential quantum measurements. }\label{fig:induced.boolean}
\end{figure}
\subsection{Problems of Interest}
In this paper, we are interested in the properties of the induced Boolean network dynamics. 
Particularly, we aim to address the following questions:
\begin{itemize}
\item[(i)] How can we represent the state transition of the $\mathbf{x}(t)$?
\item[(ii)] When and how can we characterize the basic properties of $\mathbf{x}(t)$ as a Markov chain, e.g., absorbing states, reducibility and periodicity, communication classes?
\item[(iii)] Can we establish a clear relationship between the quantum interaction structure encoded in the $\mathscr{L}$, and structures in the state space of $\mathbf{x}(t)$?
\end{itemize}
Answers to these questions will add to the understandings of the behaviors of open quantum systems in the presence of sequential measurements. 

\section{State Transition Representation}\label{sec4}
In this section, we establish an explicit representation of the state transition of the chain $\{\mathbf{x}(t)\}_{t=0}^\infty$. Such a representation is certainly non-unique, and we choose to carry out the analysis under the following standard realification of the master equation (\ref{eq:hybrid.dynamic}) (cf. e.g., \cite{WangPRA}).

Denote $N=2^n$. Let there be an orthonormal basis ${\bm\sigma}=\{\sigma_k\}_{k=1}^{N^2}$ for Hermitian operators on $\Hn$ by $\sigma_k=\lambda_{pq}$, $k=p+(q-1)N$ and $1\leq p<q\leq N$, where
\begin{align*}
&\lambda_{pq}= \frac{1}{\sqrt{2}}(\braket{p}{q}+\braket{q}{p}),\\
&\lambda_{qp}= \frac{1}{\sqrt{2}}(-i\braket{p}{q}+i\braket{q}{p}),\\
&\lambda_{pp}= \frac{1}{\sqrt{p+p^2}}\left(\sum_{k=1}^p\braket{k}{k}-p\braket{p+1}{p+1}\right).
\end{align*}
Under the basis ${\bm\sigma}$, $\rho$ is represented as a real vector $\mathbf{r}=(r_1, \dots, r_{N^2})^\top\in\bbR^{N^2}$
$$
\rho = \sum_{k=1}^{N^2}r_k\sigma_k= \sum_{k=1}^{N^2}\tr(\rho\sigma_k)\sigma_k.
$$
 
Then the Lindblad master equation \eqref{eq:lindblad} can be equivalently expressed as a real differential equation
\begin{align}\label{eq:lindblad-real}
\dot{\vr}=\left(\vL+\sum_{d}\vD^{(d)}\right)\vr := \vW\vr,
\end{align}
where $\vL$, $\vD^{(d)}\in\bbR^{N^2\times N^2}$ with entries
\begin{subequations}
\begin{align}\label{subeq:lmn}
L_{mn}&=\tr(\imath\sfH[\sigma_m,\sigma_n]),\\\label{subeq:dmn}
D_{mn}^{(d)}&=\tr(\sfV_d^\dag\sigma_m \sfV_d\sigma_n)-\frac{1}{2}\tr(\sfV_d^\dag \sfV_d\{\sigma_m ,\sigma_n\}).
\end{align}
\end{subequations}
\subsection{Transition Matrix}

Let $\pzcV:=\{1,\dots,N\}$. We introduce two mappings: 
\begin{itemize}
\item[(i)] $\btoi{\cdot}:\{0,1\}^n\rightarrow\pzcV$, where $\btoi{i_1\cdots i_n}=\sum_{k=1}^n i_k2^{n-k}+1$; 

\medskip

\item[(ii)]
$\itob{\cdot}:\pzcV\rightarrow\{0,1\}^n$ with $\itob{i}=[i_1\dots i_n]$ satisfying $i=\sum_{k=1}^n i_k2^{n-k}+1$.
\end{itemize}

Let $\sfM_{\itob{i}}:=\mathsf{P}_{i_1}\otimes\dots\otimes \mathsf{P}_{i_n}$ denote the projector onto the eigenspace generated by $\bra{v_{i_1}\cdots v_{i_n}}$ for $i_k\in\{0,1\}, k=1,\dots,n$. Upon measuring the network state $\rho$, the probability of observing $\itob{i}$ is given by
$$p(\itob{i})=\tr(\sfM_{\itob{i}}\rho).$$
Given that the outcome $\itob{i}$ occurred, the qubit network state immediately after the measurement is
$$\rho_\rmp= \braket{v_{i_1}\cdots v_{i_n}}{v_{i_1}\cdots v_{i_n}}.$$
 Then $\sfM_{\itob{i}}$ is expressed under the basis ${\bm\sigma}$ as
$$
\sfM_{\itob{i}}=\sum_{k=1}^{N^2} \theta_{i_k}\sigma_k.
$$
 Denote $\theta_i=[\theta_{i_1},\dots,\theta_{i_{N^2}}]^\top$, $i\in\pzcV$. Let $\Theta=[\theta_1,\dots,\theta_N]\in\mathbb{R}^{N^2\times N}$. 

The following theorem presents an explicit representation of the state transition characterization for the induced Boolean series $\{\mathbf{x}(t)\}_{t=0}^{\infty} $.

\begin{theorem}\label{thm:transition}
	 Along the quantum system (\ref{eq:hybrid.dynamic})--(\ref{eq:hybird.initial}), the induced Boolean network dynamics $\{\mathbf{x}(t)\}_{t=0}^{\infty} $ form a stationary Markov chain over the state space $\{0,1\}^n$, whose state transition matrix is described by 
	$$
	\vP_\tau=\Theta^\top e^{\mathbf{W}\tau}\Theta,
	$$
	where $[\vP_\tau]_{ij}$ is the transition probability from $\itob{i}$ to $\itob{j}$, here $i,j\in\pzcV$.
\end{theorem}

\begin{proof}
	Note that given $\mathbf{x}(t)$, $\xi_\rmp(t)$ is uniquely determined, which leads to $\xi(t+1)=e^{\scrL \tau}(\xi_\rmp(t))$ along \eqref{eq:hybrid.dynamic}. Therefore, the distribution of $\mathbf{x}(t+1)$, as the outcome of measuring $\xi_\rmp(t+1)$ along $\sfM^{\otimes n}$, depends only on $\mathbf{x}(t)$. This immediately implies that $\{\mathbf{x}(t)\}_{t=0}^{\infty} $ is Markovian.
	
	Next, we show the state transition matrix of $\{\mathbf{x}(t)\}_{t=0}^{\infty}$ is $\vP_\tau=\Theta^\top e^{\mathbf{W}\tau}\Theta$ by computing its each entry $[\vP_\tau]_{ij}$.
	Let $\mathbf{x}(t)=\itob{i},~i\in\pzcV$. Then the post-measurement state at step $t$ is
	$$
	\begin{aligned}
	\xi_\rmp(t)&=\braket{v_{i_1}\cdots v_{i_n}}{v_{i_1}\cdots v_{i_n}}\\
	&=\braket{v_{i_1}}{v_{i_1}}\otimes\dots\otimes\braket{v_{i_n}}{v_{i_n}}\\
	&=\mathsf{P}_{i_1}\otimes\dots\otimes \mathsf{P}_{i_n}\\
	&=\sfM_{\itob{i}}=\sum_{k=1}^{N^2}\theta_{i_k}\sigma_k.
	\end{aligned}
	$$
We now proceed to compute $\xi(t+1)$ via the real differential equation \eqref{eq:lindblad-real} of \eqref{eq:hybrid.dynamic}.
	From \eqref{eq:lindblad-real}, the coordinate vector of $\xi(t+1)$ is $e^{\mathbf{W}\tau}\theta_i$, which leads to
	$$
	\xi(t+1) = \sum_{k=1}^{N^2} [e^{\mathbf{W}\tau}\theta_i]_k\sigma_k.
	$$
	The probability of observing $\itob{j}$, $j\in\pzcV$ upon measuring $\xi(t+1)$ along $\sfM^{\otimes n}$ is
	$$
	\begin{aligned}
	\mathbb{P}(\mathbf{x}(t+1)=\itob{j}|\mathbf{x}(t)=\itob{i})
	&=\tr(\sfM_{\itob{j}}\xi(t+1))\\
	&=\tr\left(\left(\sum_{k=1}^{N^2}\theta_{j_k}\sigma_k\right)\left(\sum_{l=1}^{N^2}[e^{\mathbf{W}\tau}\theta_i]_l\sigma_l\right)\right)\\
	&=\sum_{k=1}^{N^2}\theta_{j_k}[e^{\mathbf{W}\tau}\theta_i]_k\\
	&=\theta_j^\top e^{\mathbf{W}\tau}\theta_i.
	\end{aligned}
	$$
	We can equivalently write it in a compact matrix form, which yields $\vP_\tau=\Theta^\top e^{\mathbf{W}\tau}\Theta$ exactly.
	
	This completes the proof.
%
%
\end{proof}

\section{Relaxing Quantum Dynamics}\label{sec5}
In this section, we focus on the case where the semigroup $\{e^{\scrL s}\}_{s\ge 0}$ from (\ref{eq:hybrid.dynamic}) is relaxing in the sense of Definition \ref{def1}, i.e., there exists a unique $\rho_\star$ such that
	\begin{align}\label{eq:11}
	\lim_{s\rightarrow\infty}e^{\scrL s}(\rho(0))=\rho_\star
	\end{align}
for all $\rho(0)$. There are many examples of practical quantum multi-level systems that are indeed relaxing \citep{WangPRA}. 

The following result shows that as the interval between two consecutive measurements grows, the state-transition matrix $\mathbf{P}_\tau$ tends to be close to a rank-one matrix when the semigroup $\{e^{\scrL s}\}_{s\ge 0}$ is relaxing. 
\begin{proposition}
	Suppose the semigroup $\{e^{\scrL s}\}_{s\ge 0}$ from (\ref{eq:hybrid.dynamic}) is relaxing. Then 	 for the induced Boolean network dynamics $\{\mathbf{x}(t)\}_{t=0}^{\infty} $ of the quantum system (\ref{eq:hybrid.dynamic})--(\ref{eq:hybird.initial}), there exists $\theta_{\star}\in\mathbb{R}^{N^2}$ such that
	$$
	\lim_{\tau\rightarrow\infty}\mathbf{P}_\tau=\Theta^\top (\mathbf{1}_N^\top\otimes\theta_{\star}).
	$$
\end{proposition}
\begin{proof} Let $\theta_0$ and $\theta_\star$ be the coordinates of $\rho(0)$ and $\rho_\star$ under the basis ${\bm\sigma}$, respectively. Then the semigroup $\{e^{\scrL s}\}_{s\ge 0}$ being relaxing implies that 
	$$
	\lim_{s\rightarrow\infty}e^{\mathbf{W}s}\theta_0=\theta_\star. 
	$$
As a result, invoking the representation of $\mathbf{P}_\tau$ from Theorem \ref{thm:transition}, we obtain 
	$$
	\begin{aligned}
	\lim_{\tau\rightarrow\infty}\mathbf{P}_\tau &= \lim_{\tau\rightarrow\infty}\Theta^\top e^{\mathbf{W}\tau}\Theta\\
	&= \lim_{\tau\rightarrow\infty}\Theta^\top e^{\mathbf{W}\tau}[\theta_1,\dots,\theta_{N}]\\
	& = \Theta^\top [\theta_\star,\dots,\theta_\star]\\
	& = \Theta^\top (\mathbf{1}_N^\top\otimes\theta_\star).
	\end{aligned}
	$$
	We have now completed the proof. 
\end{proof}

Next, we investigate the structural properties of the induced Boolean network dynamics $\{\mathbf{x}(t)\}_{t=0}^{\infty} $ when the semigroup $\{e^{\scrL s}\}_{s\ge 0}$ is relaxing. Recall that a state in a Markov chain is an absorbing state if it is not possible to leave whenever this chain arrived at this state. We present the following result. 
\begin{theorem}\label{thm2}
	Suppose the semigroup $\{e^{\scrL s}\}_{s\ge 0}$ from (\ref{eq:hybrid.dynamic}) is relaxing with a unique steady state $\rho_\star$. Then 	 for the induced Boolean network dynamics $\{\mathbf{x}(t)\}_{t=0}^{\infty} $ from the quantum system (\ref{eq:hybrid.dynamic})--(\ref{eq:hybird.initial}), the following statements hold. 
	\begin{itemize}
	\item[(i)]If $\rho_\star\in\{{\sfM_{\itob{i}}}\}_{i=1}^N$ and $\tau$ is sufficiently large, then the chain $\{\mathbf{x}(t)\}_{t=0}^{\infty}$ has a unique absorbing state. 
	\item[(ii)] If $\tr(\sfM_{\itob{i}}\rho_\star)>0$ for all $i=1,\dots,N$ and $\tau$ is sufficiently large, then $\{\mathbf{x}(t)\}_{t=0}^{\infty}$ is irreducible and aperiodic. 
	\end{itemize}
\end{theorem}
\begin{proof}
	With $\theta_{\star}$ being the coordinate of $\rho_\star$ under ${\bm\sigma}$, it is easy to verify that 
	$e^{\vW\tau}(\theta_\star)=\theta_\star$ for all $\tau$.
		
	(i) Suppose $\rho_\star\in\{\sfM_{\itob{i}},i\in\pzcV\}$. Then there exists $i_\star\in\pzcV$ such that $\rho_\star=\sfM_{\itob{i_\star}}$. From Theorem~\ref{thm:transition}, we know that the transition probability from $\itob{i}$ to $\itob{i_\star}$ is
	\begin{align*}
	\mathbb{P}\left(\vx(t+1)=\itob{i_\star}|\vx(t)=\itob{i}\right) 
	 = \theta_{\star}^\top e^{\vW\tau}\theta_i.
	\end{align*}
As a result, we can conclude that 
	\begin{align*}
	\lim_{\tau\rightarrow\infty}\mathbb{P}\left(\vx(t+1)=\itob{i_\star}|\vx(t)=\itob{i}\right) 
	= \lim_{\tau\rightarrow\infty}\theta_{\star}^\top e^{\vW\tau}\theta_i
	= \theta_{\star}^\top\theta_{\star} 
	> 0, 
	\end{align*}
	which implies that $$
	\mathbb{P}\left(\vx(t+1)=\itob{i_\star}|\vx(t)=\itob{i}\right)>0$$
	 for all $\itob{i}$ with sufficiently large $\tau$. Also, there holds for all $\tau>0$ that 
	\begin{align*}
	\mathbb{P}\left(\vx(t+1)=\itob{j}|\vx(t)=\itob{i_\star}\right)
	= \theta_{j}^\top e^{\vW\tau}\theta_{\star}
	 = \theta_{j}^\top \theta_{\star}
	 = 0 	
	 \end{align*}
	for any $j\in\pzcV\setminus\{i_\star\}$. Consequently, $i_\star$ is the unique absorbing state of the Markov chain $\{\mathbf{x}(t)\}_{t=0}^{\infty}$.

	(ii) Suppose $\tr(\sfM_{\itob{i}}\rho_\star)>0$ for all $i\in\pzcV$. Then 
\begin{align*}
\lim_{\tau\rightarrow\infty}\mathbb{P}\big(\vx(t+1)=\itob{i}\big|\vx(t)=\itob{j}\big)
= \theta_i^\top\theta_\star
=\tr(\sfM_{\itob{i}}\rho_\star) 
> 0
\end{align*}
for all $i,j\in\pzcV$. Similarly, this implies that for sufficiently large $\tau$,
	$$
	\mathbb{P}\big(\vx(t+1)=\itob{j}\big|\vx(t)=\itob{i}\big)>0 
	$$ 
	for all $i,j\in\pzcV$. It is then straightforward to verify that $\{\mathbf{x}(t)\}_{t=0}^{\infty}$ is irreducible and aperiodic.
\end{proof}
 
 \begin{remark}
 In Theorem \ref{thm2}, we require $\tau$ to be sufficiently large. This is due to the fact that the conclusions are built on the observation that $\rho(s)$ will be close to $\rho_\star$ as time increases. Since for relaxing open quantum systems, there holds at an exponential rate that (\cite{WangPRA})
 $$ 
 \lim_{s\to \infty}\rho(s)= \rho_\star,
 $$
 it is possible to determine an upper bound for the required $\tau$ in Theorem \ref{thm2}. 
 \end{remark}
For relaxing open quantum dynamics, the equilibrium $\rho_\star$ is in general a mixed state~\citep{WangPRA}. Consequently, unless the master equation has a very special structure, $\rho_\star \in \{{\sfM_{\itob{i}}}\}_{i=1}^N$ will not hold. Moreover, note that with relaxing quantum dynamics, in Theorem \ref{thm2}, the condition that $\tr(\sfM_{\itob{i}}\rho_\star)>0$ for all $i\in\pzcV$ holds in the generic sense, i.e., for any semigroup $\{e^{\scrL s}\}_{s\ge 0}$ from (\ref{eq:hybrid.dynamic}) that is relaxing, there holds $\tr(\sfM_{\itob{i}}\rho_\star)>0$ for all $i$ for almost all $\mathsf{M}^{\otimes n}$. This suggests that $\{\mathbf{x}(t)\}_{t=0}^{\infty}$ in general will be irreducible and aperiodic for $\tau$ sufficiently large. Then the following result establishes an asymptotic expression for the expected post-measurement quantum state $\mathbb{E}(\xi_\rmp(t))$. 
\begin{theorem}
	Suppose the induced Boolean network dynamics $\{\mathbf{x}(t)\}_{t=0}^{\infty} $ from the quantum system (\ref{eq:hybrid.dynamic})--(\ref{eq:hybird.initial}) admits a state-transition matrix $\mathbf{P}_\tau$ that is irreducible and aperiodic. Let $\pi_\tau^\ast$ be the corresponding stationary distribution. Then there holds for all $\rho(0)$ that 
	$$
	\lim_{t\rightarrow\infty}\mathbb{E}(\xi_\rmp(t))=
	\sum_{i=1}^{N} \left[\pi_\tau^\ast\right]_i (\theta_i\otimes I)\begin{bmatrix}
	\sigma_1 \\ \vdots \\ \sigma_{N^2}
	\end{bmatrix}. 
	$$
\end{theorem}
\begin{proof} Recall that the coordinate of $\rho(0)$ under ${\bm\sigma}$ is denoted as $\theta_0$. Then the probability distribution of $\mathbf{x}(0)$ is 
	\begin{align*}
	\pi_0&=\left[\bbP(\vx(0)=\itob{1}),\dots,\bbP(\vx(0)=\itob{N})\right]\\
	&=[\theta_1^\top \theta_0,\dots,\theta_N^\top \theta_0],
	\end{align*}
	which satisfies 
	$\sum_{k=1}^{N} [\pi_0]_k=1$. This implies that 
	\begin{align*}
	\lim_{t\rightarrow\infty}\mathbb{P}(\mathbf{x}(t)=[i_1\dots i_n])
	&= \lim_{t\rightarrow\infty}\left[\pi_0\mathbf{P}_\tau^t\right]_{\btoi{i_1\dots i_n}}\\
	&= \left[\pi_0\lim_{t\rightarrow\infty}\mathbf{P}_\tau^t\right]_{\btoi{i_1\dots i_n}}\\
	&= \left[\pi_0 \mathbf{1}_N\pi_\tau\right]_{\btoi{i_1\dots i_n}}\\
	&= \left[\pi_\tau^\ast\right]_{\btoi{i_1\dots i_n}},
	\end{align*}
	where $\mathbf{1}_N=[1,\dots,1]^\top$.
We can then further conclude from the distribution of $\xi_\rmp(t)$ that
	\begin{align*}
	\lim_{t\rightarrow\infty}\mathbb{E}(\xi_\rmp(t))
	&= \lim_{t\rightarrow\infty}\sum_{i_1\dots i_n}\mathbb{P}(\mathbf{x}(t)=[i_1\dots i_n])\braket{v_{i_1}\cdots v_{i_n}}{v_{i_1}\cdots v_{i_n}}\\
	&=\sum_{i_1,\dots,i_n}\left[\pi_\tau^\ast \right]_{\btoi{i_1\dots i_n}}\braket{v_{i_1}\cdots v_{i_n}}{v_{i_1}\cdots v_{i_n}}\\
	&	=\sum_{i} \left[\pi_\tau^\ast\right]_i (\theta_i\otimes I)\begin{bmatrix}
	\sigma_1 \\ \vdots \\ \sigma_{N^2}
	\end{bmatrix}.
	\end{align*}
	We have now completed the proof of the desired conclusion. 
\end{proof}

Note that $\mathbb{E}(\xi_\rmp(t))$ is the classical expectation of $\xi(t)$, where the probability measure arises from the random measurement outcomes. Therefore, the implication is that the quantum pre-measurement state $\xi(t)$ tends to be stationary in terms of its measurement statistics. However, $\xi(t),t=0,1,2,\dots$ will undergo ergodic recursions among $N$ different states in the set
$$
\Big\{e^{-\mathscr{L}\tau} ({\sfM_{\itob{i}}}) \Big\}_{i=1}^N,
$$
instead of being convergent. 

\section{Non-Relaxing Quantum Dynamics: Quantum Consensus Networks}\label{sec6}
In this section, we turn our attention to non-relaxing quantum dynamics (\ref{eq:hybrid.dynamic}). It is clear that various types of master equations could lead to non-relaxing quantum dynamics. Instead of looking into the general form of (\ref{eq:hybrid.dynamic}), we discuss the quantum network dynamics (\ref{eq:hybrid.dynamic})--(\ref{eq:hybird.initial}) under the so-called consensus master equation \citep{shi2016}. On one hand, the consensus master equation \citep{shi2016} is potentially an important class of open quantum networks, analogous to the classical Laplacian consensus dynamics~\citep{Magnus}. On the other hand, how the simple yet rich structural effect in the consensus master equation \citep{shi2016} affects the sequential measurement outcomes is an interesting point for investigation. 

\subsection{Consensus Master Equation}
 A permutation of the qubit set $\mathrm{V}=\{1,\dots,n\}$ is a bijective map from $\mathrm{V}$ onto itself. We denote by $\chi$ such a permutation.
 Particularly, a permutation $\chi$ is called a swapping between $j$ and $k$, denoted by $\chi_{jk}$, if $\chi(j)=k$, $\chi(k)=j$, and $\chi(l)=l, l\in \mathrm{V}\setminus\{j,k\}$.
 The set of all permutations of $\mathrm{V}$ forms a group, called the $n$'th permutation group and denoted by $\Upsilon_n=\{\chi\}$. There are $n!$ elements in $\Upsilon_n$.

\begin{definition}
Let $\chi\in \Upsilon_n$. We define the unitary operator $\mathsf{U}_\chi$ over $\Hn$ induced by $\chi$, by
$$
\mathsf{U}_\chi \big(|q_1\rangle \otimes \dots \otimes |q_n\rangle \big)= |q_{\chi(1)}\rangle \otimes \dots \otimes |q_{\chi(n)}\rangle,
$$
where for $i=1,\dots,n$. Similarly, with slight abuse of notation, we define the action of $\chi$ over $\{0,1\}^{\otimes n}$ by
$$
\chi\big([i_1\dots i_n]\big)=[i_{\chi_1}\dots i_{\chi_n}]
$$ 
where $i_k\in\{0,1\}$ for all $k=1,\dots,n$. 
\end{definition}

Let the operator $\mathsf{U}_{\chi_{_{jk}}}$ be denoted as $\mathsf{U}_{jk}$ for the ease of presentation. Let the network interaction structure be described by an undirected and, without loss of generality, connected graph $\mathrm{G}=(\mathrm{V}, \mathrm{E})$. The so-called quantum consensus master equation is described by \citep{shi2016} 
\begin{align}\label{sys10}
\dot{\rho}(s)=\mathscr{L}(\rho(s))=\sum_{\{j,k\}\in \mathrm{E}} \alpha_{jk} \Big(\mathsf{U}_{jk}\rho \mathsf{U}_{jk}^\dag -\rho\Big), 
\end{align}
 where $\alpha_{jk}>0$ represents the weight of link $\{j,k\}$.

 Define an operator over the density operators of $\Hn$, $\mathscr{P}_\ast$, by
\begin{align}
\label{quantumconsensus}
\mathscr{P}_\ast (\rho)= \frac{1}{n!} \sum_{\chi \in \Upsilon_n} \mathsf{U}_\chi \rho \mathsf{U}^\dag_\chi.
\end{align}
It is known that when the graph $\mathrm{G}$ is connected, along the equation (\ref{sys10}) there holds
 \begin{align}
 \lim_{s\to\infty} \rho(s)= \mathscr{P}_\ast (\rho_0)
 \end{align}
with $\rho(0)=\rho_0$. Clearly the master equation (\ref{sys10}) is not relaxing as the limiting point depends on the initial quantum state. 

\subsection{State Transitions}
We are now in a place to study the quantum network dynamics (\ref{eq:hybrid.dynamic})--(\ref{eq:hybird.initial}) when the continuous quantum dynamics (\ref{eq:hybrid.dynamic}) is described by (\ref{sys10}).

Recall that $\sfP_m=\braket{v_m}{v_m}$ is the projector onto the eigenspace generated by the eigenvector $\bra{v_m} \in \HH_{2}$ with eigenvalue $\lambda_m$, $m\in\{0,1\}$. Let $\big\{|i^\sharp\rangle\langle j^\sharp|\big\}_{i,j=1}^{N^2}$ be a basis of $\calL(\HH_N)$, where by definition 
$$
i^\sharp:=v_{i_1}v_{i_2}\dots v_{i_n}
$$ 
with $\itob{i}=[i_1,\dots i_n]$. According to the definition of $\mathsf{U}_\chi$, we can verify that 
\begin{align}
\mathsf{U}_\chi \big( |i^\sharp\rangle\big)= \big| v_{i_{\chi(1)}}v_{i_{\chi(2)}}\dots v_{i_{\chi(n)}} \big\rangle.
\end{align}
As a result, under the basis of $\big\{|i^\sharp\rangle \big\}_{i=1}^{N}$, the matrix representation of $\mathsf{U}_\chi$, denoted ${U}_\chi$, is a real permutation matrix for any $\chi\in\Upsilon_n$. Similarly, we denote ${U}_{jk}$ as the matrix representation of the operator $\mathsf{U}_{jk}$ under the basis $\big\{|i^\sharp\rangle \big\}_{i=1}^{N}$.
\begin{definition}{\normalfont\citep{shi2016}}
The quantum Laplacian of $\mathrm{G}$ is defined as 
$$
L_{\rm q}(\mathrm{G}):=-\sum_{\{j,k\}\in\mathrm{E}}\alpha_{jk}\big( {U}_{jk} \otimes {U}_{jk}-I \big). 
$$
\end{definition}

Let ${e}_i$ denote the $N\times 1$ unit vector with the $i$-th entry being one and all other entries being zero. Then we can establish the following result. 

\begin{proposition}\label{prop2}
Consider (\ref{eq:hybrid.dynamic})--(\ref{eq:hybird.initial}) with (\ref{eq:hybrid.dynamic}) being described by the quantum consensus master equation (\ref{sys10}) under qubit interaction graph $\mathrm{G}$. Define
$$
E_N:=\left[e_1\otimes e_1 ,\dots,e_N \otimes e_N\right].
$$ 
Then there holds for the $\{\mathbf{x}(t)\}_{t=0}^{\infty} $ under the measurement $\mathsf{M}^{\otimes n}$ that 
$$
\mathbf{P}_\tau =E_N^\top e^{-\tau L_{\rm q}(\mathrm{G})} E_N. 
$$
\end{proposition}
Proposition \ref{prop2} shows that the exponential of the quantum Laplacian directly characterizes the state transition matrix of the induced probabilistic Boolean dynamics $\{\mathbf{x}(t)\}_{t=0}^{\infty} $. 
The proof of Proposition \ref{prop2} follows from a similar process as the proof of Theorem \ref{thm:transition}, where we only need to notice the following two points:
\begin{itemize}
\item[(i)] Under the basis $\big\{|i^\sharp\rangle\langle j^\sharp|\big\}_{i,j=1}^{N^2}$ for $\calL(\HH_N)$, the density operator $\rho$ can be represented as 
$$
\rho=\sum_{i,j=1}^{N^2} \rho_{ij}|i^\sharp\rangle\langle j^\sharp|. 
$$
As a result, the coordinate vector of $\rho$ is precisely ${\rm vec}([\rho_{ij}])$. Then the consensus master equation (\ref{sys10}) can be written as 
\begin{align}
\frac{d}{ds}{\rm vec}([\rho_{ij}(s)])= -L_{\rm q}(\mathrm{G}){\rm vec}([\rho_{ij}(s)]). 
\end{align}
Here for any given matrix $M\in \mathbb{C}^{m\times n}$, the vectorization of $M$, denoted by ${\rm vec}(M)$,
is the $mn\times 1$ column vector $([M]_{11}, \dots, [M]_{m1}, \dots, [M]_{1n},\dots, [M]_{mn})^\top$ \citep{Horn}.
\item[(ii)] Under the basis $\big\{|i^\sharp\rangle\langle j^\sharp|\big\}_{i,j=1}^{N^2}$, there holds 
$$
\sfM_{\itob{i}}=\mathsf{P}_{i_1}\otimes\dots\otimes \mathsf{P}_{i_n}= \big|\btoi{i_1\cdots i_n}^\sharp \big \rangle \big\langle \btoi{i_1\cdots i_n}^\sharp \big|. 
$$
\end{itemize}

We also note a critical observation from Proposition \ref{prop2} that the representation of $\mathbf{P}_\tau$ is independent of the choice of $\mathsf{M}$ in the network observable $\mathsf{M}^{\otimes n}$, in sharp contrast with Theorem \ref{thm:transition}. The consensus master equation has inherent symmetry, where the selection of the basis $\big\{|i^\sharp\rangle\langle j^\sharp|\big\}_{i,j=1}^{N^2}$ preserves the representation $\mathbf{P}_\tau$ under different measurement bases. Additionally, it is obvious that $\mathbf{P}_\tau$ is a symmetric matrix.

\subsection{Communication Classes}

For the Markov chain $\{\mathbf{x}(t)\}_{t=0}^{\infty} $, a state $[p_1 \dots p_n]$ in its state space is said to be {\it accessible} from state $[q_1 \dots q_n]$ if there is a nonnegative integer $t$ such that $$
\mathbb{P}\big(\mathbf{x}_t=[p_1 \dots p_n]\ \big|\ \mathbf{x}_0=[q_1 \dots q_n]\big)>0. 
$$ It is termed that $[p_1 \dots p_n]$ {\it communicates} with state $[q_1 \dots q_n]$ if $[p_1 \dots p_n]$ and $[q_1 \dots q_n]$ are accessible from each other.
This communication relationship forms an equivalence relation among the states in $\{0,1\}^n$. The equivalence classes of this relation are called {\it communication classes} of the chain $\{\mathbf{x}(t)\}_{t=0}^{\infty} $. The following theorem provides a full characterization to the communication classes of $\{\mathbf{x}(t)\}_{t=0}^{\infty} $ under the consensus master equation.

\begin{theorem}\label{thm4}
Consider (\ref{eq:hybrid.dynamic})--(\ref{eq:hybird.initial}) with (\ref{eq:hybrid.dynamic}) being described by the quantum consensus master equation (\ref{sys10}) under qubit interaction graph $\mathrm{G}$. Then the following statements hold for the $\{\mathbf{x}(t)\}_{t=0}^{\infty} $ under the measurement $\mathsf{M}^{\otimes n}$.
\begin{itemize}
\item[(i)] There are $n+1$ different communication classes. 
\item[(ii)] For any $g=[g_1,\dots,g_n]\in\{0,1\}^n$, the communication class containing $g$ is given by
\begin{align}
\mathcal{C}_g=\Big\{\chi (g):\chi\in\Upsilon_n \Big\}. 
\end{align}
\item[(iii)] The number of states in $\mathcal{C}_g$ is $$
\binom{n}{|g|}=\frac{n!}{|g|!(n-|g|)!}
$$ 
with $|g|=\sum_{i=1}^n g_i$. 
\end{itemize}
\end{theorem}
\begin{proof} We first establish two technical lemmas as preliminaries of the proof. The proofs of the lemmas can be found in the appendix. 
\begin{lemma}\label{lem1}
For any $g=[g_1,\dots,g_n]\in\{0,1\}^n$, the space 
\begin{align}
\HH_g={\rm span}\Big\{\mathsf{U}_{\chi} |g\rangle \langle g| \mathsf{U}_{\chi}^\dag:\chi\in\Upsilon_n \Big\}
\end{align}
is invariant under the quantum consensus master equation (\ref{sys10}). 
\end{lemma}

\begin{lemma}\label{lem2}
Let $L(\mathrm{G})$ be the classical Laplacian associated with an undirected and connected graph $\mathrm{G}=(\mathrm{V},\mathrm{E})$. Then 
\begin{align}
e^{-\tau L(\mathrm{G})}>0
\end{align}
holds entrywise, i.e., $e^{-\tau L(\mathrm{G})}$ is a positive matrix, for all $\tau>0$. 
\end{lemma}

We are now in a place to prove the theorem. Suppose $\mathbf{x}(t)=g=[g_1,\dots,g_n]$. Then the post-measurement state
at step $t$ is $$
\xi_{\rm p}(t)= |g\rangle \langle g|. 
$$
For the continuous time interval $s\in[t\tau,(t+1)\tau)$, Lemma \ref{lem1} suggests that there holds
\begin{align}\label{101}
\rho(s)\in \HH_g
\end{align}
for all $s\in[t\tau,(t+1)\tau)$. As a result, we can establish the following two aspects. 

\begin{itemize}
\item [(a)] The accessible states from $g$ are in the set 
$
\mathcal{C}_g.
$
In fact, if $g'\notin\mathcal{C}_g $, then 
\begin{align}
{\rm Tr} \Big(|g'\rangle \langle g'| \mathsf{U}_{\chi} |g\rangle \langle g| \mathsf{U}_{\chi}^\dag \Big)=0
\end{align}
for all $\chi\in\Upsilon_n$ as $|g'\rangle$ and $\mathsf{U}_{\chi} |g\rangle$ must be two orthogonal elements in the basis $\big\{|i^\sharp\rangle\big\}_{i=1}^{N}$. Then from (\ref{101}), 
$$
{\rm Tr} \Big(|g'\rangle \langle g'|\rho(((t+1)\tau)^-) \Big)=0. 
$$Therefore, any accessible state from $g$ belongs to 
$\mathcal{C}_g$.

\item[(b)] Any state in $\mathcal{C}_g$ is indeed accessible from $g$. We can establish this point by focusing our analysis on the quantum Laplacian $L_{\rm q}(\mathrm{G})$. This quantum Laplacian, is the Laplacian of a generalized graph $\mathpzc{G}=(\mathpzc{V},\mathpzc{E})$ \citep{shi2016}. Let $\mathpzc{V}^\ast$ be defined by 
$$
\mathpzc{V}^\ast:=\Big\{\btoi{\chi(g)}:\chi\in\Upsilon_n \Big\}.
$$
The induced graph of $
\mathpzc{V}^\ast$ over $\mathpzc{G}$, $\mathpzc{G}|_{\mathpzc{V}^\ast}$, is obviously a connected subgraph from the quantum consensus master equation (\ref{sys10}).

 Then for the quantum Laplacian $L_{\rm q}(\mathrm{G})$ associated with the qubit interaction graph $\mathrm{G}=(\mathrm{V},\mathrm{E})$, there holds 
\begin{align}\label{102}
&\mathbb{P}(\mathbf{x}(t+1)= \chi(g)|\mathbf{x}(t)=g)\nonumber\\
&={\rm Tr} \Big(|\chi(g)\rangle \langle \chi(g)|\rho(((t+1)\tau)^-) \Big)\nonumber\\
&=(\btoi{\chi(g)}_N \otimes \btoi{\chi(g)}_N)^\top e^{-\tau L_{\rm q}(\mathrm{G})} (\btoi{g}_N \otimes \btoi{g}_N)\nonumber\\
&= \Big[e^{-\tau L(\mathpzc{G}|_{\mathpzc{V}^\ast}) }\Big]_{\btoi{\chi(g)}\btoi{g}}
\nonumber\\
&>0
\end{align}
where the first identity follows from the quantum measurement postulate, the second identity follows from the facts that 
$$
{\rm Tr}(A^\top B)=\sum_{i,j=1}^N A_{ij}B_{ij}= {\rm vec}(A)^\top{\rm vec}(B)
$$
for two matrices $A$ and $B$, and that $|\chi(g)\rangle \langle \chi(g)|$ is symmetric as a projector, and the last inequality holds from Lemma \ref{lem2}. This shows that the states in
$
\mathcal{C}_g
$
are all accessible from $g$. 
\end{itemize} 

Combining the two aspects, we can conclude that for the chain $\{\mathbf{x}(t)\}_{t=0}^{\infty} $, the communication class that a state $g$ is in, precisely $\mathcal{C}_g$. This establishes the statement (ii).

A further look of $\mathcal{C}_g$ leads to the immediate observation that for any $g\in\{0,1\}^{\otimes n}$, there holds
\begin{align} 
\mathcal{C}_g=\Big\{[i_1 \dots i_n]:\ \sum_{k=1}^n i_k =|g|\Big\}.
\end{align}
The statements (i) and (iii) thus follow from a basic analysis of combinatorics. We have now completed the proof of the desired theorem. 
\end{proof}

Theorem \ref{thm4} is closely related to the notion of generalized graph of the quantum interaction graph introduced in \cite{shi2016}. The generalized graph is the graph that is consistent with the quantum Laplacian, where for an $n$-qubit network, its generalized graph contains $N^2=4^n$ nodes. Particularly, the communication class $\mathcal{C}_g$ essentially coincides with the connected components of the $N$ nodes in the generalized graph corresponding to the diagonal entry of the network density operator. 

\subsection{Example}
We now present a concrete example as an illustration of the established results for quantum consensus networks with sequential measurements. We consider three qubits indexed by $1$, $2$, and $3$. The qubit interaction graph $\mathrm{G}=(\mathrm{V},\mathrm{E})$ is assumed to be a path graph as shown in Figure~\ref{fig1}. 


\begin{figure}[h] 
\centering
\includegraphics[width=3.2cm]{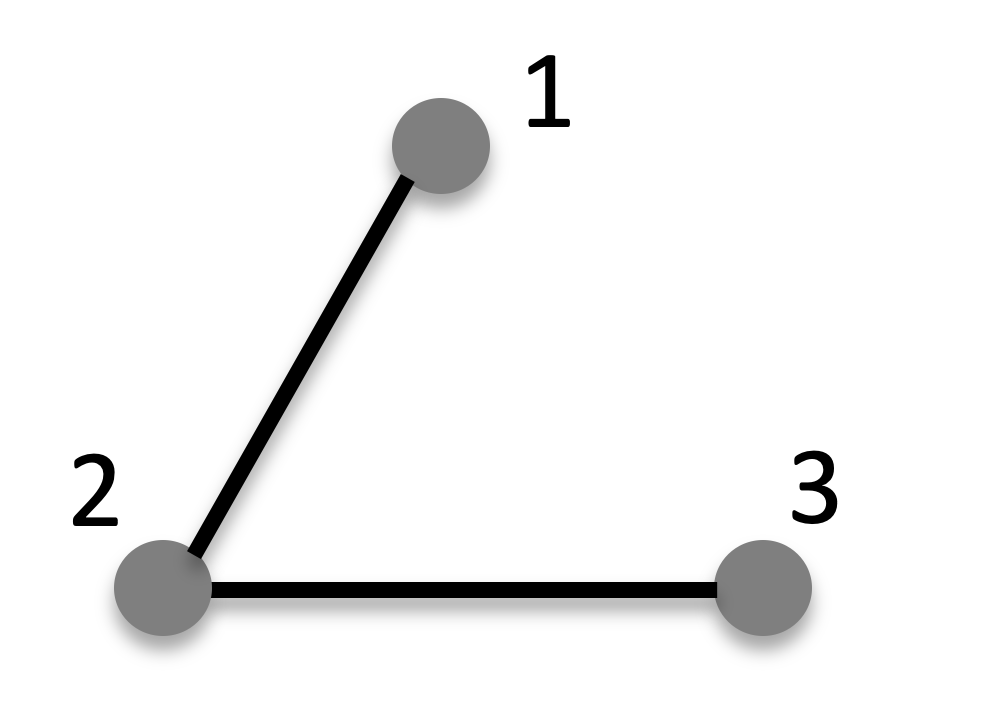}
\caption{The interaction graph for the three-qubit network. }\label{fig1}
\end{figure}

The quantum Laplacian $L_{\rm q}(\mathrm{G})$ for this graph is a $64\times 64$ matrix. Let the measurement $\mathsf{M}$ be taken under the standard computational basis, without loss of generality, i.e.,
\begin{align}
\mathsf{M}=\lambda_0 \braket{0}{0}+\lambda_1 \braket{1}{1},
\end{align}
and the resulting network measurement is $\mathsf{M}^{\otimes 3}$. Let the continuous quantum state follow the evolution described by the quantum consensus master equation (\ref{sys10}) with two swapping operators $\mathsf{U}_{12}$ and $\mathsf{U}_{23}$ as specified from the interaction graph $\mathrm{G}$. Let the measurement $\mathsf{M}^{\otimes 3}$ be carried out periodically with inter-measurement time $\tau=1$. The measurement outcome for the $s$'th measurement is recorded as $\mathbf{x}(t)\in\{0,1\}^3$. Then we can verify the following aspects. 

(i) The state transition matrix of the chain $\mathbf{x}(t)$ is given by
\begin{align}
\mathbf{P}_{\tau} =E_8^\top e^{-L_{\rm q}(\mathrm{G})} E_8\nonumber
&= \begin{pmatrix}
1 & 0 & 0 & 0 & 0 & 0 & 0 & 0\\
0 & 0.525 & 0.317 & 0 & 0.158 & 0 & 0 & 0\\
0 & 0.317& 0.366& 0& 0.317& 0 & 0 & 0\\
0 & 0 & 0 & 0.525 & 0 & 0.317 & 0.158 & 0\\
0 & 0.158 & 0.317 & 0 & 0.525 & 0 & 0 & 0\\
0 & 0 & 0 & 0.317 & 0 & 0.366& 0.317 & 0\\
0 & 0 & 0 & 0.158 & 0 & 0.317 & 0.525 & 0\\
0 & 0 & 0& 0 & 0 & 0 & 0 & 1 
\end{pmatrix}.\nonumber
\end{align}

(ii) Let $g=[001]$. Note that $\btoi{001}=2$. From the second row of $\mathbf{P}_{\tau}$, clearly the three nonzero entries are $[\mathbf{P}_{\tau}]_{22}>0$, $[\mathbf{P}_{\tau}]_{23}>0$ and $[\mathbf{P}_{\tau}]_{25}>0$. Consequently, the 
 states that are accessible from $g$ are $[001]=\itob{2},$ $[010]=\itob{3}$, and $[100]=\itob{5}$.
 
 On the other hand, we can verify directly that 
 \begin{align}
\mathcal{C}_g=\big\{\chi (g):\chi\in\Upsilon_3 \big\}=\big\{[001],[010],[100]\big\}
\end{align} 
which is consistent with the communication class that we established directly from $\mathbf{P}_{\tau}$. This is a validation of Theorem \ref{thm4}.(ii) above.

\begin{figure}[h] 
\centering
\begin{tikzpicture}[auto,node distance=2cm]
\tikzstyle{state}=[draw=blue,thick,minimum size=6mm]
\node[state] (000) {$[000]$};
\node[state,right of=000] (001) {$[001]$};
\node[state,right of=001] (010) {$[010]$};
\node[state,right of=010] (011) {$[011]$};
\node[state,below of=000,yshift=.5cm] (100) {$[100]$};
\node[state,right of=100] (101) {$[101]$};
\node[state,right of=101] (110) {$[110]$};
\node[state,right of=110] (111) {$[111]$};
\draw[<->,>=stealth,color=red,thick] (001) -- (010);
\draw[<->,>=stealth,color=red,thick] (010) -- (100);
\draw[<->,>=stealth,color=red,thick] (100) -- (001);
\draw[<->,>=stealth,color=cyan,thick] (101) -- (110);
\draw[<->,>=stealth,color=cyan,thick] (110) -- (011);
\draw[<->,>=stealth,color=cyan,thick] (011) -- (101);
\end{tikzpicture}
\caption{The state transition map from $\mathbf{P}_{\tau}$ for the measurement outcomes $\mathbf{x}(t)$. }\label{fig3}
\end{figure}
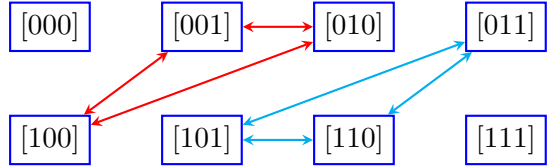

(iii) We can also establish from $\mathbf{P}_{\tau}$ (see the resulting state transition map in Figure \ref{fig3}) that the communication classes of $\mathbf{x}(t)$ are 
\begin{align*}
&\big\{ [000] \big\}; \nonumber\\
&\big\{[001],[010],[100]\big\}; \nonumber\\
&\big\{[011],[101],[110]\big\}; \nonumber\\
&\big\{[111]\big\}.
\end{align*}
The number of the communication classes and the size of each communication class are clearly consistent with Theorem \ref{thm4}.(i) and Theorem \ref{thm4}.(iii).

\section{Conclusions}\label{sec7}
 Open quantum networks, as a proven resource for universal quantum computation, are networked quantum subsystems such as qubits with the interconnections established by local environments. Their state evolutions can be described by structured master equations, and in the presence of sequential quantum measurements, the network states undergo random jumps with the measurement outcomes form a probabilistic Boolean network. We showed that the the state transition of the random measurement outcomes can be explicitly represented from the master equation. It was also shown that structural properties including absorbing states, reducibility, and periodicity for the induced Boolean dynamics can be made clear directly when the quantum dynamics is relaxing. For quantum consensus networks as a type of non-relaxing open quantum network dynamics, we showed that the communication classes of the measurement-induced Boolean networks arise from the quantum Laplacian of the underlying interaction graph.

\section*{Appendix} 
\subsection*{A. Proof of Lemma \ref{lem1}}
Let $\rho(s)$ be evolving along the quantum consensus master equation (\ref{sys10}) with $\rho(0)\in\HH_g$. We carry out the analysis under the basis $\big\{|i^\sharp\rangle\big\}_{i=1}^{N}$ for $\HH_N$, and therefore each density operator can be understood as its matrix representation with slight abuse of notation.

Introduce
\begin{align}
h_g(s)=\frac{1}{2}\big\| \rho(s) \big\|_{\HH_g}^2
\end{align}
where $\|a\|_{\HH_g}$ stands for distance between a point $a$ and the space ${\HH_g}$ under the $2$-norm. Let $\rho(0)=g$. Noting the fact that $\HH_g$ is linear subspace in $\mathcal{L}(\HH_N)$, we obtain 
\begin{align}
&\frac{d}{ds} h_g(s)\Big|_{s=0}\nonumber\\
&=\Big\langle {\rm vec}(\rho(s))- {\rm vec}(\mathscr{P}_{\HH_g}(\rho(s))), {\rm vec}(\dot\rho(s))\Big\rangle\Big|_{s=0} \nonumber\\
&=\Big\langle {\rm vec}(\rho(s))- {\rm vec}(\mathscr{P}_{\HH_g}(\rho(s))), \nonumber\\
& \quad {\rm vec}(\sum_{\{j,k\}\in \mathrm{E}} \alpha_{jk} \Big(\mathsf{U}_{jk}\rho(s) \mathsf{U}_{jk}^\dag -\rho(s)\Big))\Big\rangle\Big|_{s=0}\nonumber\\
&=0
\end{align}
where $\mathscr{P}_{\HH_g}(\cdot)$ is the projection onto the space $\HH_g$, and the last equality holds from the fact that
\begin{align}
\mathsf{U}_{jk}g \mathsf{U}_{jk}^\dag\in \HH_g
\end{align}
according to the definition of $\HH_g$. Therefore, $h_g(s)\equiv 0$ if $h_g(0)=0$. This proves that $\HH_g$ is invariant and the desired lemma holds. 
\subsection*{B. Proof of Lemma \ref{lem2}}
It suffices to prove that 
$
e^{-\tau L(\mathrm{G})} h_i$
is a positive vector for each $i=1,\dots,n$, where $ h_i$ is the $n\times 1$ unit vector with the $i$-th entry being one. This is equivalent to $$
\mathbf{y}(\tau)=(y_1(\tau),\dots,y_n(\tau))^\top >0
$$ when $\mathbf{y(s)}$ is evolving along the Laplacian consensus 
\begin{align}\label{100}
\dot{\mathbf{y}}(s) =- L(\mathrm{G}) \mathbf{y}(s)
\end{align} 
starting from $\mathbf{y}(0)= h_i$. It is well known that $ \mathbf{y}(s)\geq 0$ along (\ref{100}) for all $s\geq 0$ since the convex hull of $y_1(0),\dots,y_n(0),$ denoted ${\rm co}(y_1(0),\dots,y_n(0))$, is invariant under the Laplacian consensus dynamics. 

Next, we establish the following claim. 

{\it Claim.} Let $\{i,j\}\in \mathrm{E}$. Then $y_j(\tau)>0$ for sufficiently small $\tau>0$.

This claim can be easily established through the fact that 
\begin{align*}
\frac{d}{ds}{y}_j(s) \Big|_{s=0}&=\sum_{k=1}^n [L(\mathrm{G})]_{kj}(y_k(s)-y_j(s)) \Big|_{s=0}\nonumber\\
&=[L(\mathrm{G})]_{ij}\nonumber\\
&>0. 
\end{align*}
This shows that each of $i$'s neighbor will hold a positive state during $[0,\tau_0)$ for some small $\tau_0$. Carrying out the similar analysis recursively to node $i$'s neighbors' neighbor, etc., with connectivity we conclude that each node $i$ will hold a positive state during $[0,\tau_0)$ for some small $\tau_0$.

Finally, once $\mathbf{y}(\tau_0)>0$, there holds $ \mathbf{y}(s)> 0$ for all $s\geq \tau_0$, because, again, ${\rm co}(y_1(\tau_0),\dots,y_n(\tau_0))$ is invariant under the Laplacian consensus dynamics. This proves the desired lemma.


\begin{thebibliography}{22}
	\expandafter\ifx\csname natexlab\endcsname\relax\def\natexlab#1{#1}\fi
	\providecommand{\url}[1]{\texttt{#1}}
	\providecommand{\href}[2]{#2}
	\providecommand{\path}[1]{#1}
	\providecommand{\DOIprefix}{doi:}
	\providecommand{\ArXivprefix}{arXiv:}
	\providecommand{\URLprefix}{URL: }
	\providecommand{\Pubmedprefix}{pmid:}
	\providecommand{\doi}[1]{\href{http://dx.doi.org/#1}{\path{#1}}}
	\providecommand{\Pubmed}[1]{\href{pmid:#1}{\path{#1}}}
	\providecommand{\bibinfo}[2]{#2}
	\ifx\xfnm\relax \def\xfnm[#1]{\unskip,\space#1}\fi
	\bibitem[{Altafini(2002)}]{Altafini2002}
	\bibinfo{author}{Altafini, C.} (\bibinfo{year}{2002}).
	\newblock \bibinfo{title}{Controllability of quantum mechanical systems by root
		space decomposition of su(n)}.
	\newblock {\it \bibinfo{journal}{J. Mathematical Physics}\/},  {\it
		\bibinfo{volume}{43}\/}, \bibinfo{pages}{2051--2062}.
	\bibitem[{Altafini \& Ticozzi(2012)}]{Altafini-Ticozzi-2012}
	\bibinfo{author}{Altafini, C.}, \& \bibinfo{author}{Ticozzi, F.}
	(\bibinfo{year}{2012}).
	\newblock \bibinfo{title}{Modeling and control of quantum systems: an
		introduction}.
	\newblock {\it \bibinfo{journal}{IEEE Trans. Automatic Control}\/},  {\it
		\bibinfo{volume}{57}\/}, \bibinfo{pages}{1898--1917}.
	\bibitem[{Belavkin(1999)}]{Belavkin1979}
	\bibinfo{author}{Belavkin, V.~P.} (\bibinfo{year}{1999}).
	\newblock \bibinfo{title}{Optimal measurement and control in quantum dynamical
		systems}.
	\newblock {\it \bibinfo{journal}{Rep. Math. Phys.}\/},  {\it
		\bibinfo{volume}{43}\/}, \bibinfo{pages}{405--425}.
	\newblock \bibinfo{note}{Preprint No. 411, Inst. of Phys., Nicolaus Copernicus
		University, Torun’, February 1979.}
	\bibitem[{Blok et~al.(2014)Blok, Bonato, Markham, Twitchen, Dobrovitski \&
		Hanson}]{QubitFeedback2014}
	\bibinfo{author}{Blok, M.~S.}, \bibinfo{author}{Bonato, C.},
	\bibinfo{author}{Markham, M.~L.}, \bibinfo{author}{Twitchen, D.~J.},
	\bibinfo{author}{Dobrovitski, V.~V.}, \& \bibinfo{author}{Hanson, R.}
	(\bibinfo{year}{2014}).
	\newblock \bibinfo{title}{Manipulating a qubit through the backaction of
		sequential partial measurements and real-time feedback}.
	\newblock {\it \bibinfo{journal}{Nature Physics}\/},  {\it
		\bibinfo{volume}{10}\/}, \bibinfo{pages}{189}.
	\bibitem[{Breuer \& Petruccione(2002)}]{Breuer-Petruccione-2002}
	\bibinfo{author}{Breuer, H.-P.}, \& \bibinfo{author}{Petruccione, F.}
	(\bibinfo{year}{2002}).
	\newblock {\it \bibinfo{title}{The Theory of Open Quantum Systems}\/}.
	\newblock (\bibinfo{edition}{1st} ed.).
	\newblock \bibinfo{publisher}{Oxford University Press}.
	\bibitem[{Durrett(2005)}]{durrett}
	\bibinfo{author}{Durrett, R.} (\bibinfo{year}{2005}).
	\newblock {\it \bibinfo{title}{Probability: Theory and Examples}\/}.
	\newblock \bibinfo{address}{Pacific Grove, CA, USA}:
	\bibinfo{publisher}{Brooks/Cole}.
	\bibitem[{Godsil \& Royle(2001)}]{godsil}
	\bibinfo{author}{Godsil, C.}, \& \bibinfo{author}{Royle, G.}
	(\bibinfo{year}{2001}).
	\newblock {\it \bibinfo{title}{Algebraic Graph Theory}\/}.
	\newblock \bibinfo{address}{New York}: \bibinfo{publisher}{Springer-Verlag}.
	\bibitem[{Horn \& Johnson(1985)}]{Horn}
	\bibinfo{author}{Horn, R.~A.}, \& \bibinfo{author}{Johnson, C.~R.}
	(\bibinfo{year}{1985}).
	\newblock {\it \bibinfo{title}{Matrix Analysis}\/}.
	\newblock \bibinfo{publisher}{Cambridge University Press}.
	\bibitem[{James et~al.(2008)James, Nurdin \& Petersen}]{ian2008}
	\bibinfo{author}{James, M.~R.}, \bibinfo{author}{Nurdin, H.~I.}, \&
	\bibinfo{author}{Petersen, I.~R.} (\bibinfo{year}{2008}).
	\newblock \bibinfo{title}{$h_\infty$ control of linear quantum stochastic
		systems}.
	\newblock {\it \bibinfo{journal}{IEEE Trans. Automatic Control}\/},  {\it
		\bibinfo{volume}{53}\/}, \bibinfo{pages}{1787--1803}.
	\bibitem[{Lindblad(1976)}]{Lindblad1976}
	\bibinfo{author}{Lindblad, G.} (\bibinfo{year}{1976}).
	\newblock \bibinfo{title}{On the generators of quantum dynamical semigroups}.
	\newblock {\it \bibinfo{journal}{Comm. Math. Phys.}\/},  {\it
		\bibinfo{volume}{48}\/}, \bibinfo{pages}{119--130}.
	\bibitem[{Mazzarella et~al.(2015)Mazzarella, Sarlette \& Ticozzi}]{Ticozzi2015}
	\bibinfo{author}{Mazzarella, L.}, \bibinfo{author}{Sarlette, A.}, \&
	\bibinfo{author}{Ticozzi, F.} (\bibinfo{year}{2015}).
	\newblock \bibinfo{title}{Consensus for quantum networks: from symmetry to
		gossip iterations}.
	\newblock {\it \bibinfo{journal}{IEEE Trans. Automatic Control}\/},  {\it
		\bibinfo{volume}{60}\/}, \bibinfo{pages}{158--172}.
	\bibitem[{Mesbahi \& Egerstedt(2010)}]{Magnus}
	\bibinfo{author}{Mesbahi, M.}, \& \bibinfo{author}{Egerstedt, M.}
	(\bibinfo{year}{2010}).
	\newblock {\it \bibinfo{title}{Graph Theoretic Methods in Multiagent
			Networks}\/}.
	\newblock \bibinfo{publisher}{Princeton University Press}.
	\bibitem[{Nielsen \& Chuang(2010)}]{Nielsen}
	\bibinfo{author}{Nielsen, M.~A.}, \& \bibinfo{author}{Chuang, I.~L.}
	(\bibinfo{year}{2010}).
	\newblock {\it \bibinfo{title}{Quantum Computation and Quantum Information}\/}.
	\newblock \bibinfo{publisher}{Cambridge University Press}.
	\bibitem[{Pechen et~al.(2006)Pechen, Il’in, Shuang \& Rabitz}]{Rabitz-PRA}
	\bibinfo{author}{Pechen, A.}, \bibinfo{author}{Il’in, N.},
	\bibinfo{author}{Shuang, F.}, \& \bibinfo{author}{Rabitz, H.}
	(\bibinfo{year}{2006}).
	\newblock \bibinfo{title}{Quantum control by von {Neumann} measurements}.
	\newblock {\it \bibinfo{journal}{Physical Review A}\/},  {\it
		\bibinfo{volume}{74}\/}, \bibinfo{pages}{052102}.
	\bibitem[{Schirmer \& Wang(2010)}]{WangPRA}
	\bibinfo{author}{Schirmer, S.~G.}, \& \bibinfo{author}{Wang, X.}
	(\bibinfo{year}{2010}).
	\newblock \bibinfo{title}{Stabilizing open quantum systems by {Markovian}
		reservoir engineering}.
	\newblock {\it \bibinfo{journal}{Physical Review A}\/},  {\it
		\bibinfo{volume}{81}\/}, \bibinfo{pages}{062306}.
	\bibitem[{Shi et~al.(2016)Shi, Dong, Petersen \& Johansson}]{shi2016}
	\bibinfo{author}{Shi, G.}, \bibinfo{author}{Dong, D.},
	\bibinfo{author}{Petersen, I.~R.}, \& \bibinfo{author}{Johansson, K.~H.}
	(\bibinfo{year}{2016}).
	\newblock \bibinfo{title}{Reaching a quantum consensus: Master equations that
		generate symmetrization and synchronization}.
	\newblock {\it \bibinfo{journal}{IEEE Trans. Automatic Control}\/},  {\it
		\bibinfo{volume}{61}\/}, \bibinfo{pages}{374--387}.
	\bibitem[{Shmulevich et~al.(2002)Shmulevich, Dougherty, Kim \&
		Zhang}]{Probabilistic-Boolean-Network}
	\bibinfo{author}{Shmulevich, I.}, \bibinfo{author}{Dougherty, E.~R.},
	\bibinfo{author}{Kim, S.}, \& \bibinfo{author}{Zhang, W.}
	(\bibinfo{year}{2002}).
	\newblock \bibinfo{title}{Probabilistic {B}oolean networks: a rule-based
		uncertainty model for gene regulatory networks}.
	\newblock {\it \bibinfo{journal}{Bioinformatics}\/},  {\it
		\bibinfo{volume}{2}\/}, \bibinfo{pages}{261--274}.
	\bibitem[{Ticozzi \& Viola(2008)}]{Ticozzi-Viola-2008}
	\bibinfo{author}{Ticozzi, F.}, \& \bibinfo{author}{Viola, L.}
	(\bibinfo{year}{2008}).
	\newblock \bibinfo{title}{Quantum {M}arkovian subsystems: Invariance,
		attractivity and control}.
	\newblock {\it \bibinfo{journal}{IEEE Trans. Automatic Control}\/},  {\it
		\bibinfo{volume}{53}\/}, \bibinfo{pages}{2048--2063}.
	\bibitem[{Ticozzi \& Viola(2009)}]{Ticozzi-Viola-2009}
	\bibinfo{author}{Ticozzi, F.}, \& \bibinfo{author}{Viola, L.}
	(\bibinfo{year}{2009}).
	\newblock \bibinfo{title}{Analysis and synthesis of attractive quantum
		{M}arkovian dynamics}.
	\newblock {\it \bibinfo{journal}{Automatica}\/},  {\it \bibinfo{volume}{45}\/},
	\bibinfo{pages}{2002--2009}.
	\bibitem[{Verstraete et~al.(2009)Verstraete, Wolf \& Cirac}]{NaturePhysiscs}
	\bibinfo{author}{Verstraete, F.}, \bibinfo{author}{Wolf, M.~M.}, \&
	\bibinfo{author}{Cirac, J.~I.} (\bibinfo{year}{2009}).
	\newblock \bibinfo{title}{Quantum computation and quantum-state engineering
		driven by dissipation}.
	\newblock {\it \bibinfo{journal}{Nature Physics}\/},  {\it
		\bibinfo{volume}{5}\/}, \bibinfo{pages}{633--636}.
	\bibitem[{Wang et~al.(2012)Wang, Pemberton-Ross \& Schirmer}]{Wang2012}
	\bibinfo{author}{Wang, X.}, \bibinfo{author}{Pemberton-Ross, P.}, \&
	\bibinfo{author}{Schirmer, S.~G.} (\bibinfo{year}{2012}).
	\newblock \bibinfo{title}{Symmetry and subspace controllability for spin
		networks with a single-node control}.
	\newblock {\it \bibinfo{journal}{IEEE Trans. Automatic Control}\/},  {\it
		\bibinfo{volume}{57}\/}, \bibinfo{pages}{1945--1956}.
	\bibitem[{Wiseman \& Milburn(2010)}]{Wiseman-Milburn-2009}
	\bibinfo{author}{Wiseman, H.~M.}, \& \bibinfo{author}{Milburn, G.~J.}
	(\bibinfo{year}{2010}).
	\newblock {\it \bibinfo{title}{Quantum Measurement and Control}\/}.
	\newblock \bibinfo{address}{Cambridge, England}: \bibinfo{publisher}{Cambridge
		University Press}.
	
\end{thebibliography}
\end{document}